\documentclass{ifacconf}  

\usepackage{natbib}  

\usepackage{comment}
\usepackage{amssymb}
\usepackage{amsmath,amssymb}
\usepackage{xcolor}
\usepackage{url}
\usepackage{amsfonts}
\usepackage{booktabs}
\usepackage{multirow} 
\usepackage{multicol}
\usepackage{makecell}
\usepackage{textcomp}

\newtheorem{remark}{Remark}

\newtheorem{proposition}{Proposition}

\usepackage{graphics} 
\usepackage{graphicx}

\usepackage[rightcaption]{sidecap}
\usepackage{algorithm}
\usepackage{algpseudocode}
\usepackage{float}
\usepackage{tabularx}
\usepackage{array}
\usepackage{wrapfig}
\usepackage{enumitem}
\usepackage{fancyhdr}
\fancypagestyle{preprint}{
  \fancyhf{} 
  
  \fancyhead[C]{\fontsize{10}{12}\selectfont This work has been submitted to IFAC for possible publication}
}

\newlist{Properties}{enumerate}{2}
\setlist[Properties]{label=Property \arabic*., itemindent=*}

\makeatletter
\DeclareFontFamily{U}{tipa}{}
\DeclareFontShape{U}{tipa}{m}{n}{<->tipa10}{}
\newcommand{\arc@char}{{\usefont{U}{tipa}{m}{n}\symbol{62}}}%

\newcommand{\arc}[1]{\mathpalette\arc@arc{#1}}

\newcommand{\arc@arc}[2]{%
  \sbox0{$\m@th#1#2$}%
  \vbox{
    \hbox{\resizebox{\wd0}{\height}{\arc@char}}
    \nointerlineskip
    \box0
  }%
}
\makeatother

\usepackage[utf8]{inputenc}

\begin{document}

\begin{frontmatter}

\title{Multi-UAV Path Following using Vector-Field Guidance} 

\author[First]{Gautam Kumar} 
\author[Second]{Amit Shivam} 
\author[Third]{Ashwini Ratnoo}

\address[First]{Department of Aerospace Engineering, Indian Institute of Science Bangalore, India 560012 (e-mail: gautamkumar1@iisc.ac.in)}
\address[Second]{Department of Electrical and Computer Engineering, University of Porto, Porto, Portugal,
		(e-mail: amitshivam1407@alum.iisc.ac.in)}
\address[Third]{Department of Aerospace Engineering, Indian Institute of Science Bangalore, India 560012, (e-mail: ratnoo@iisc.ac.in)}

\begin{abstract}  

This paper presents a decentralized, collision-free framework for path following guidance of multiple uncrewed aerial vehicles (UAVs), while maintaining uniform spacing along a reference path. A vector field-based guidance law is employed to drive each UAV toward the reference path. A rotational repulsion mechanism, utilizing relative distance and bearing between UAVs, is proposed to avoid collisions during convergence to the path, and an inter-UAV spacing error-based velocity control law is presented to achieve uniform separation along the path. Analytical guarantees are established for collision avoidance and for the convergence of the inter-UAV spacing errors to zero, ensuring uniform separation along the path. Numerical simulations demonstrate the efficacy of the proposed method.

\end{abstract}

\end{frontmatter}

\thispagestyle{preprint}


\section{Introduction}

With the rapid deployment of uncrewed aerial vehicles (UAVs) in civilian and industrial applications, such as aerial surveillance, package delivery, and infrastructure inspection, there has emerged a need for safe and robust trajectory planning methods. In a structured common airspace, UAVs are often required to fly along predefined routes, necessitating the design of collision avoidance strategies. In such scenarios, maintaining a uniform inter-UAV separation along the desired path is equally important, as it facilitates deterministic UAV traffic flow, efficient airspace utilization, and simplified conflict resolution at intersections (\cite{nagrare2024intersection,he2025air}).

Achieving safe UAV operations and uniform inter-UAV separation necessitates coordinated control strategies that account for the inter-UAV interactions. In this context, formation-based approaches have been widely explored, as they enable structured motion through prescribed relative geometries among agents. A class of approaches considers circular or closed-path following, in which inter-UAV spacing and collision avoidance are maintained through phase-angle coordination along the desired closed path (\cite{li2021fully,mei2024enhanced}). However, in a structured airspace, the desired routes are often characterized by open paths, such as sinusoidal curves and straight line paths (\cite{quan2023distributed}). In these scenarios, the phase-angle-based approaches are not directly applicable as the open paths can no longer be parameterized using a continuous angular variable.

To address multi-UAV path-following along a designated non-closed path, various formation control strategies have been proposed in the literature. In \cite{ghommam2009coordinated}, a Serret-Frenet formulation-based path-following control law coupled with a speed adjustment mechanism is developed for UAVs tracking a common reference in a synchronized manner. That approach is centralized, which may limit its applicability to large swarms. A hybrid control strategy is proposed in \cite{chen2021coordinated}, in which the path-following error for an agent is shown to continuously decrease within its coordination set, defined by the agent's dynamics and motion constraints. The approaches in \cite{ghommam2009coordinated,chen2021coordinated} lack explicit treatment of inter-UAV collision avoidance guarantees. In \cite{prodan2010collision,chen2025distributed}, coordinated control strategies for following continuous curved paths ensure collision avoidance by incorporating inter-UAV distance constraints into Model Predictive Control frameworks. These approaches require solving high-dimensional optimization problems at each control step, thereby restricting utility for a large UAV swarm due to the high computational burden. While the approach in \cite{zuo2021coordinated} achieves the tracking of distinct parameterized paths for different agents and maintains temporal coordination between agents, it does not address convergence to a common path with uniform inter-agent spacing.

Vector field-based methods have been widely employed in multi-UAV systems for various applications, including traversal in obstacle-laden environments (\cite{hegde2016multi,pothen2017curvature}), evader capture (\cite{viet2012univector}), and patrolling missions (\cite{gao2018cooperative}), owing to their advantages of being computationally light and robust to external disturbances. The guiding vector field-based path following method in \cite{yao2021distributed} ensures path-following along a reference closed path for multiple robots but lacks explicit inter-UAV collision avoidance. Motivated by the need for scalable and collision-free path-following of a common path by multiple UAVs, this work presents a decentralized coordination strategy that enables multiple UAVs to converge to a common open path by leveraging an arcsine-based vector-field guidance law, while maintaining safe inter-UAV separation and uniform spacing along the path. The contributions of this work are as follows:
\begin{enumerate}
    \item A decentralized rotational repulsion mechanism is integrated within the vector field guidance framework, which enables UAVs to perform repulsive action in the presence of potential collisions, thereby facilitating collision avoidance.
    \item A longitudinal velocity control law, based on inter-UAV spacing error, is designed to achieve and maintain equal spacing along the desired path.
\end{enumerate}

The remainder of the paper is organized as follows: Section \ref{sec:prob} introduces the problem scenario. The proposed path following framework is discussed in Section \ref{sec:pathfollow}, followed by collision avoidance and uniform spacing between UAVs along the reference path in Section \ref{sec:collision_equispacing}. Simulation results are presented in Section \ref{sec:simulation}, followed by concluding remarks in Section \ref{sec:conclusion}.

\section{Problem Statement}\label{sec:prob}

Consider a team of $N$ UAVs operating in a planar environment. The position of the $i$th UAV is denoted by $\mathbf{x}_i$, where $\mathbf{x}_i = [x_i,\, y_i]^\top \in \mathbb{R}^2$, and its motion is governed by the kinematic model as
\begin{align}
    \dot{x}_i = v_i\cos(\psi_i),\qquad
    \dot{y}_i = v_i\sin(\psi_i),\qquad
    \dot{\psi}_i = \omega_i
\end{align}
where $\psi_i$ is the heading angle and $(v_i,\omega_i)$ are the linear and angular velocity input of the $i$th UAV. Let $\mathcal{P}\in \mathbb{R}^2$ be a desired reference path defined in the $xy$-plane. In this work, two representative paths are considered (Fig. \ref{fig:prob_statement}), 
\begin{align}
\text{(i)}~ \mathcal{P}_1 &: x = 0, \label{eq:path1} \\
\text{(ii)}~ \mathcal{P}_2 &: x = A \sin(k y), \label{eq:path2}
\end{align}
corresponds to straight line and sinusoidal curve paths, respectively. Here $A$ and $k$ are the amplitude and frequency parameters of path $\mathcal{P}_2$. The objective is to design a decentralized control law to deduce $(v_i,\omega_i)$ for each UAV such that:
\begin{enumerate}
    \item the UAVs converge to the reference path $\mathcal{P}_m\; (m = \{1,2\})$, 
    \begin{align}
        &\lim_{t \to \infty} \text{dist}(\mathbf{x}_i(t), \mathcal{P}_m) = 0
    \end{align}
    where $\text{dist}(\mathbf{x}_i, \mathcal{P}_m)$ denotes the minimum Euclidean distance between the $i$th UAV and $\mathcal{P}_m$.
    \item a minimum safety distance $d_{\mathrm{safe}}>0$ is maintained between UAVs, 
    \begin{align}
       & \|\mathbf{x}_i(t) - \mathbf{x}_j(t)\| > d_{\text{safe}},\\& \forall t\geq 0, \;\; i,j = 1,2,\ldots, N \text{ and } i\neq j.
    \end{align}
    \item the UAVs must maintain equal spacing of length $d_{\mathrm{eq}}$ along the path, 
    \begin{align}
        \lim_{t \to \infty}\big(s_i(t) - s_{i+1}(t) - d_{\text{eq}}\big) \to 0, \label{eq:equispacing}
    \end{align}
    where $s_i$ represents the arc length parameter of the $i$th UAV along the path.
\end{enumerate}
Further, the following assumptions have been considered in this work:
\begin{enumerate}[label = (A\arabic*)]
    \item The environment is obstacle-free.
    \item The path information is known to the UAVs.
    \item The UAVs can communicate their position and velocity information with each other.
\end{enumerate}
\begin{figure}[!hbt]
    \centering
    \includegraphics[width=0.75\linewidth]{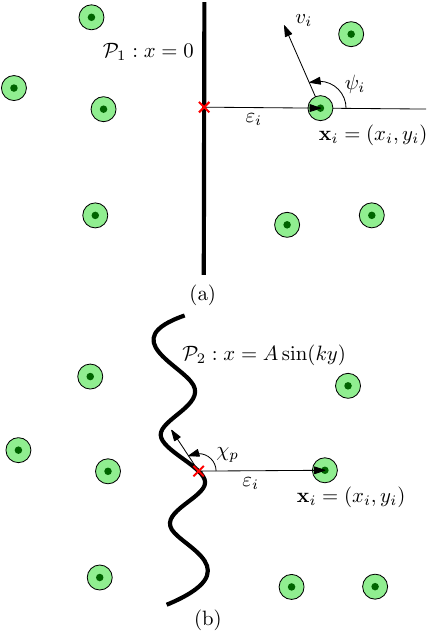}
    \caption{Path following scenario: (a) straight line path, and (b) sinusoidal path following.}
    \label{fig:prob_statement}
\end{figure}

\section{Path Following Methodology}\label{sec:pathfollow}

In this section, a vector field-based guidance methodology is discussed to guide an individual UAV along the desired reference paths in \eqref{eq:path1} and \eqref{eq:path2}.

The desired heading angle $\psi_i^{\mathrm{des}}$ for the $i$th UAV is composed of two terms (\cite{shivam2023arcsine}):
\begin{equation}\label{eq:psi_des}
\psi_i^{\mathrm{des}} =
\begin{cases}
\chi_i^p - \chi_i^o, & \varepsilon \le 0,\\[0.1em]
\chi_i^p + \chi_i^o, & \varepsilon > 0,
\end{cases}
\end{equation}
where $\chi_i^p $ is the local path-tangential direction, and $\chi_i^o$ is an arcsine-shaped offset function of the cross-track error $\varepsilon_i$, given by
\begin{align}
\chi_i^p = \tan^{-1}\!\left(\frac{\mathrm{d}y_i}{\mathrm{d}x_i}\right), ~ \chi_i^o =  \dfrac{\pi}{2} - \sin^{-1}\!\left(\frac{1}{1 + k_g \varepsilon_i^2}\right)\label{eq:chio_def} ,
\end{align}
where $k_g > 0$ is the guidance gain.

\subsection{Straight Line Path Following}
The desired path, in this case, is $\mathcal{P}_1$, where the cross-track error $\varepsilon_i= x_i$ (Fig. \ref{fig:prob_statement}a). The local path-tangential direction in \eqref{eq:chio_def}, $\chi_i^p = \dfrac{\pi}{2}$. Using \eqref{eq:psi_des} and \eqref{eq:chio_def}, the behavior of $\psi_i^{\mathrm{des}}$ as a function of $\varepsilon_i$ can be characterized as follows:
\begin{align}\label{eq:st_analysis}
   & \lim_{\varepsilon_i \to -\infty} \psi_i^{\mathrm{des}} = 0, \quad
     \lim_{\varepsilon_i \to \infty} \psi_i^{\mathrm{des}} = \pi, \\
   & \text{and} \quad \psi_i^{\mathrm{des}} = \frac{\pi}{2} \;\; \text{at} \;\; \varepsilon_i = 0.
\end{align}
Using \eqref{eq:st_analysis}, when the UAV is located far on the left side of the desired path $\mathcal{P}_1$, that is, $\varepsilon_i \ll 0$, the heading angle $\psi_i^{\mathrm{des}} \to 0$, which guides the UAV rightward toward the path. 
Conversely, when $\varepsilon_i \gg 0$, $\psi_i^{\mathrm{des}} \to \pi$, which commands the UAV a leftward motion towards $\mathcal{P}_1$. At the desired path, that is, $\varepsilon_i = 0$, the heading angle of the UAV aligns with the path tangent, that is, $\psi_i^{\mathrm{des}} = \pi/2$, ensuring straight motion along the reference line. The Lyapunov-based analysis in \cite{shivam2023arcsine} shows asymptotic convergence of the cross-track error $\varepsilon_i \to 0$ for any initial condition, thus guaranteeing convergence of the UAV to the desired path $\mathcal{P}_1$.

\subsection{Sinusoidal Path Following}

For the sinusoidal path scenario, the reference path $\mathcal{P}_2$ in \eqref{eq:path2} is considered. Here, the cross-track error is defined as the lateral deviation from the desired path $\mathcal{P}_2$, that is,
\begin{equation}
\varepsilon_i = x_i - A \sin(k y_i).
\end{equation}
The local path-tangential direction, $\chi_p$ is given by
\begin{align}\label{eq:chio_sinedef}
\chi_i^p &= \tan^{-1}\!\left(\frac{\mathrm{d}y_i}{\mathrm{d}x_i}\right) = \tan^{-1}\!\left(\frac{1}{Ak\cos(ky_i)}\right) 
\end{align}
The formulation of $\psi_i^{\mathrm{des}}$ using \eqref{eq:psi_des} and \eqref{eq:chio_def} ensures a smooth variation of the desired heading angle with respect to the cross-track error $\varepsilon_i$. Accordingly, the desired heading angle $\psi_i^{\mathrm{des}}$ satisfies
\begin{align}
   & \lim_{\varepsilon_i \to -\infty} \psi_i^{\mathrm{des}} = \chi_i^p - \frac{\pi}{2},\quad \lim_{\varepsilon_i \to \infty} \psi_i^{\mathrm{des}} = \chi_i^p + \frac{\pi}{2}, \text{ and}\label{eq:left_pos} \\
   &  \text{at} \;\; \varepsilon_i = 0,\; \psi_i^{\mathrm{des}} = \chi_i^p \;.
\end{align}
From \eqref{eq:left_pos}, when the UAV is far to the left of $\mathcal{P}_2$, that is, $\varepsilon_i \ll 0$, the desired heading angle guides the UAV to take a rightward turn toward the path. Conversely, when $\varepsilon_i \gg 0$, the UAV takes a leftward turn towards $\mathcal{P}_2$. Once the UAV reaches the desired path, that is, $\varepsilon_i = 0$, the heading direction aligns with the local path tangent $\chi_p$. The asymptotic convergence to the desired path $\mathcal{P}_2$ is theoretically shown using the Lyapunov analysis in \cite{shivam2024class}.

\begin{remark}\label{rem:pathomega}
The desired heading angle $\psi_i^{\mathrm{des}}$ obtained from the guidance law provides the reference heading angle for the $i$th UAV to align with the desired path. Accordingly, the angular velocity component that drives the agent to its desired path is obtained using a proportional control law.
\begin{equation}\label{eq:omega_control}
    \omega_i^{\mathrm{path}} = k_\psi \left( \psi_i^{\mathrm{des}} - \psi_i \right),
\end{equation}
where $k_{\psi}  > 0$ is the heading control gain. 
\end{remark}

\section{Collision Avoidance and Equispacing Control}
\label{sec:collision_equispacing}

This section presents a collision avoidance mechanism and a spacing error-based velocity control law that are incorporated into the path-following guidance framework discussed in Section \ref{sec:pathfollow}. Formal guarantees are established for avoiding collisions between UAVs and for convergence on inter-UAV spacing to the desired separation along the reference path.

\subsection{Collision Avoidance Mechanism}

The relative distance and bearing angle between the $i$th and $j$th UAV are defined, respectively, as
\begin{equation}
  d_{i,j} = \|\mathbf{x}_i - \mathbf{x}_j\|, \qquad
  \beta_{i,j} = \tan^{-1}\!\left(\frac{y_j - y_i}{x_j - x_i}\right).
\end{equation}
The rotational repulsion component of the angular velocity input of the $i$th UAV is formulated as
\begin{equation}\label{eq:omega_rep}
  \omega_i^{\mathrm{rep}} =
  \begin{cases}
    \displaystyle
    -k_r \sum_{j \in \mathcal{N}_i}
    \!\left(\frac{1}{d_{i,j}} - \frac{1}{R_{s}}\right)
    \sin\!\bigl(\beta_{i,j} - \psi_i\bigr),
    \;\text{if } d_{i,j} \le R_{s},\\
    0, \hspace{5.3cm} \text{otherwise,}
  \end{cases}
\end{equation}
where $R_s$ denotes the threshold safety distance, $\mathcal{N}_i$ is the set of neighboring UAVs within the circular region of radius $R_s$ centered at $\mathbf{x}_i$, and $k_r > 0$ is the repulsion gain. This repulsive angular component in \eqref{eq:omega_rep} is composed of two terms:
\begin{enumerate}
  \item The term $\left(\dfrac{1}{d_{i,j}} - \dfrac{1}{R_{s}}\right)$ is a distance-dependent weighting factor that controls the repulsive turning rate based on the proximity between the $i$th and $j$th UAVs. At $d_{i,j} = R_{s}$, $\omega_i^{\mathrm{rep}}=0$, and grows monotonically as $d_{i,j}$ decreases below $R_{s}$, generating an angular deviation away from the neighboring agent. This inverse-distance dependence is analogous to the gradient of classical potential-field repulsion functions (\cite{rimon1990exact,khatib1986real}).

  \item The factor $\sin(\beta_{i,j} - \psi_i)$ determines the direction of rotational repulsion, that is, whether to turn clockwise or counterclockwise, ensuring that the repulsive angular component guides the UAV away from its neighboring UAV.
\end{enumerate}
\begin{remark}
    Using Remark \ref{rem:pathomega} and \eqref{eq:omega_rep}, the total angular velocity input for the $i$th UAV is
\begin{equation}\label{eq:omega_total}
  \omega_i = \omega_i^{\mathrm{path}} + \omega_i^{\mathrm{rep}}.
\end{equation}
\end{remark}

The following proposition establishes the collision avoidance property of the repulsive angular component in \eqref{eq:omega_rep} under the equal-speed assumption $v_i = v_j = v$. Additionally, the repulsion component is assumed sufficiently large to dominate the total angular velocity input in the case of a potential collision.

\begin{proposition}\label{prop:collision}
  Consider two UAVs $i$ and $j$, both employing the rotational repulsion law \eqref{eq:omega_rep} with equal nominal speeds $v_i = v_j = v$.  Then there exists a finite gain $k_r^* > 0$ such that, for all $k_r \ge k_r^*$, the separation $d_{i,j}(t)$ does not decrease monotonically to zero, that is, whenever $d_{i,j} < R_s$ and $\dot{d}_{i,j} < 0$, the repulsion induces $\ddot{d}_{i,j} > 0$, decelerating the approach and driving $\dot{d}_{i,j}$ towards positive values before the UAVs can collide.
\end{proposition}

\begin{pf}
The lead angles for the UAVs are defined as
\begin{equation}
  \phi_i \triangleq \beta_{i,j} - \psi_i, \qquad
  \phi_j \triangleq \beta_{i,j} - \psi_j.
\end{equation}
Under equal speeds $v_i = v_j = v$, the range and bearing angle dynamics are governed by
\begin{align}
  \dot{d}_{i,j} &= v\left(\cos\phi_j - \cos\phi_i\right),\;
  \dot{\beta}_{i,j}
  = \frac{v}{d_{i,j}}\left(\sin\phi_j - \sin\phi_i\right)\label{eq:betadot}.
\end{align}
In the scenario of potential collision, we have  $\dot{\psi}_k \approx \omega_k^{\mathrm{rep}}$ $(k = \{i,j\})$. Differentiating $\phi_i$ and $\phi_j$, and substituting \eqref{eq:betadot},
\begin{align}
  \dot{\phi}_i
    &= \frac{v}{d_{i,j}}\bigl(\sin\phi_j - \sin\phi_i\bigr)
       - \omega_i^{\mathrm{rep}}, \label{eq:phi_i_dot}\\[4pt]
  \dot{\phi}_j
    &= \frac{v}{d_{i,j}}\bigl(\sin\phi_j - \sin\phi_i\bigr)
       - \omega_j^{\mathrm{rep}}. \label{eq:phi_j_dot}
\end{align}
Since $\beta_{ji} = \beta_{i,j} + \pi$, it follows that
$\omega_j^{\mathrm{rep}} = -\omega_i^{\mathrm{rep}}$. Let $\Omega$ denote the common magnitude $ k_r\!\left(\frac{1}{d_{i,j}} - \frac{1}{R_s}\right)$ in $\omega_k^{\text{rep}}$ such that $\omega_i^{\mathrm{rep}} = -\Omega\sin\phi_i$ and
$\omega_j^{\mathrm{rep}} = \Omega\sin\phi_j$. Differentiating \eqref{eq:betadot},
\begin{equation}\label{eq:ddotd_start}
  \ddot{d}_{i,j}
  = v\bigl(\sin\phi_j\,\dot{\phi}_j - \sin\phi_i\,\dot{\phi}_i\bigr).
\end{equation}
Substituting \eqref{eq:phi_i_dot} and \eqref{eq:phi_j_dot} into \eqref{eq:ddotd_start},
\begin{align}
  \ddot{d}_{i,j}
  &= v\!\left[\sin\phi_j\!\left(\frac{v}{d_{i,j}}(\sin\phi_j\!-\!\sin\phi_i)
     - \Omega\sin\phi_j\right)\right.\nonumber\\
  &\qquad\left.
     - \sin\phi_i\!\left(\frac{v}{d_{i,j}}(\sin\phi_j\!-\!\sin\phi_i)
     + \Omega\sin\phi_i\right)\right]\nonumber\\[4pt]
  &=  \underbrace{-\frac{v^2}{d_{i,j}}(\sin\phi_j - \sin\phi_i)^2}_{\le\,0}
  + \underbrace{v\,\Omega\,\bigl(\sin^2\phi_i + \sin^2\phi_j\bigr)}_{\ge\,0}. \label{eq:final_ddotdij}
\end{align}
The first term is always non-positive. The second term is non-negative for $d_{i,j} \le R_s$ (since $\Omega \ge 0$). In the scenario where both UAVs are approaching each other such that neither of the UAVs' heading is exactly aligned with the line of sight, that is, $\sin\phi_i \ne 0$ or $\sin\phi_j \ne 0$, the second term is strictly positive. The critical separation $d^*_{i,j}$ at which $\ddot{d}_{i,j} =0$ is obtained from \eqref{eq:final_ddotdij} as
\begin{equation}\label{eq:dstar}
  d_{i,j}^* = R_s\!\left(1 - 
  \frac{v\,(\sin\phi_j - \sin\phi_i)^2}
       {k_r\,(\sin^2\phi_i + \sin^2\phi_j)}\right).
\end{equation}
For collision avoidance, the sign of $\dot{d}_{i,j}$ must become positive before $d_{i,j}$ reaches $d_{\mathrm{safe}}$. Since $\ddot{d}_{i,j}>0$ for $d_{i,j}<d_{i,j}^*$, the range rate $\dot{d}_{i,j}$ is increasing in that region. To ensure $d_{i,j}>d_{\mathrm{safe}}$, 
\begin{equation}\label{eq:kr_bound}
   k_r>k_r^* =
  \frac{v\, R_s}{R_s - d_{\mathrm{safe}}} \cdot
  \frac{(\sin\phi_j - \sin\phi_i)^2}{\sin^2\phi_i + \sin^2\phi_j}.
\end{equation}
For $k_r=k_r^*$, $\ddot{d}_{i,j} = 0 $ occurs at ${d}_{i,j}^* =  d_{\mathrm{safe}}$. Since $ \frac{(\sin\phi_j - \sin\phi_i)^2}{(\sin^2\phi_i + \sin^2\phi_j)}\leq 2$, a sufficient condition for $d_{i,j}^*>d_{\mathrm{safe}}$ for all geometries $(\phi_i,\phi_j)$ is $k_r>\frac{2v\, R_s}{R_s - d_{\mathrm{safe}}}$. Therefore, a sufficiently large gain $k_r$ exists such that $d_{i,j}^* > d_{\mathrm{safe}}$, which drives the range rate $\dot{d}_{i,j}$ to positive values, before $d_{i,j}=d_{\mathrm{safe}}$.
\end{pf}

\begin{remark}
  The equal-speed assumption $v_i = v_j = v$ is adopted for analytical tractability and holds once equispacing is achieved in accordance with Proposition \ref{prop:equispacing}, discussed subsequently in Section \ref{sec:eqspace}. During the transient phase in which the speeds of UAVs $i$ and $j$ differ, the repulsion gain $k_r$ must be chosen sufficiently large to compensate for the effect of the speed deviation $|v_i - v_{\mathrm{nom}}|$ and maintain the avoidance property.
\end{remark}

\subsection{Equispacing Velocity Control}\label{sec:eqspace}

To achieve a uniform spacing between the UAVs along the desired path, the longitudinal velocity of each UAV is adjusted based on the spacing error with respect to its immediate predecessor. The spacing error $\Delta_i$ for the $i$th UAV is defined as
\begin{equation}
  \Delta_i = s_i - s_{i-1} - d_{\mathrm{eq}}, \qquad i \ge 2,
\end{equation}
where $d_{\mathrm{eq}}$ and $s$ are the desired inter-UAV spacing and the path parameter, respectively. This predecessor-following spacing error follows the string topology studied in \cite{swaroop2002string}. In this work, the path parameter $s$ is defined as follows:
\begin{enumerate}
  \item For the straight line path $\mathcal{P}_1 : x = 0$,
        \begin{equation}
          s(y) = y.
        \end{equation}
  \item For the sinusoidal path $\mathcal{P}_2 : x = A\sin(ky)$,
        \begin{equation}
          s(y) = \int_{0}^{y}
          \sqrt{1 + \bigl(Ak\cos(k\tau)\bigr)^2}\,\mathrm{d}\tau.
        \end{equation}
\end{enumerate}
The longitudinal velocity control law is then expressed as
\begin{equation}\label{eq:v_control}
  v_i = v_{\mathrm{nom}} - \kappa\,\tanh(\Delta_i),
\end{equation}
where $\kappa \in (0,\,v_{\mathrm{nom}})$ is a spacing gain and $v_{\mathrm{nom}}$ is the nominal speed common to all UAVs, and the hyperbolic tangent function structure ensures 
\begin{equation}
    v_i \in (v_{\mathrm{nom}}-\kappa,\,
v_{\mathrm{nom}}+\kappa),
\end{equation}
which preserves a positive forward speed at all times. By design, the first UAV ($i=1$) has no predecessor and therefore maintains $v_1 = v_{\mathrm{nom}},\; \forall t\geq 0$.

\begin{proposition}\label{prop:equispacing}
  Consider the longitudinal velocity control law \eqref{eq:v_control} applied to UAVs $i = 2, \ldots, N$, with UAV~$1$ maintaining $v_1 = v_{\mathrm{nom}}$. Assume each UAV has converged to the desired path using path-following guidance in Section \ref{sec:pathfollow}, such that $\dot{s}_i = v_i$ for all $i$. Then the inter-UAV spacing errors $\Delta_i$ asymptotically
  converge to zero, that is,
  \begin{equation}
    \lim_{t \to \infty} \Delta_i(t) = 0, \qquad \forall\, i \ge 2.
  \end{equation}
\end{proposition}

\begin{pf}
Consider the Lyapunov candidate
\begin{equation}\label{eq:lya_Can}
  V = \frac{1}{2}\sum_{i=2}^{N} \Delta_i^2 \ge 0.
\end{equation}
The time derivative of $V$ is
\begin{equation}\label{eq:lya_der1}
  \dot{V} = \sum_{i=2}^{N} \Delta_i\,\dot{\Delta}_i
           = \sum_{i=2}^{N} \Delta_i\,(\dot{s}_i - \dot{s}_{i-1}).
\end{equation}
Under the assumption $\dot{s}_i = v_i$, substituting
\eqref{eq:v_control} gives
\begin{align}
  \dot{V}
  &= \sum_{i=2}^{N} \Delta_i\,(v_i - v_{i-1}) \nonumber\\
  &= -\kappa\sum_{i=2}^{N}
     \Delta_i\bigl(\tanh(\Delta_i) - \tanh(\Delta_{i-1})\bigr),
     \label{eq:Vdot_final}
\end{align}
where $\Delta_1 = 0$ indicates that the first UAV's speed is $v_{\mathrm{nom}}$. Since $\tanh(\cdot)$ is monotonically increasing, each term  $\Delta_i\bigl(\tanh(\Delta_i)-\tanh(\Delta_{i-1})\bigr)$ is non-negative, which gives $\dot{V} \le 0$. By LaSalle's invariance principle (\cite{khalil2002nonlinear}), all trajectories converge to the largest invariant subset of  $\{\dot{V} = 0\}$. Since each term in \eqref{eq:Vdot_final} is non-negative, $\dot{V} = 0$ implies
\begin{equation}\label{eq:lassle}
  \Delta_i\bigl(\tanh(\Delta_i) - \tanh(\Delta_{i-1})\bigr) = 0, 
  \qquad \forall\, i \ge 2.
\end{equation}
For $i = 2$, using $\Delta_1 = 0$, this reduces to
\begin{equation}
  \Delta_2\,\tanh(\Delta_2) = 0 \implies \Delta_2 = 0
\end{equation}
as $x\tanh(x) = 0$ iff $x = 0$. Assuming $\Delta_{i-1} = 0$  has been established, Eq. \ref{eq:lassle} reduces to $\Delta_i\tanh(\Delta_i) = 0$, which similarly implies $\Delta_i = 0$. Hence, proceeding recursively, $\Delta_i = 0$ for 
all $i \ge 2$. Therefore, the largest invariant set in $\{\dot{V} = 0\}$ is $\{\Delta_i = 0,\; \forall\, i \ge 2\}$, and asymptotic convergence of all spacing errors to zero follows from LaSalle's invariance principle.
\end{pf}

\begin{remark}
  The path-parameter rate $\dot{s}_i$ equals the component of velocity tangent to the path, that is, $\dot{s}_i = v_i\cos(\psi_i - \chi_i^p)$, which reduces to $v_i$ when $\mathbf{x}_i\in \mathcal{P}_m\;(m = \{1,2\})$. Therefore, the asymptotic convergence to the path established in Section \ref{sec:pathfollow} justifies the assumption $\dot{s}_i=v_i$ used in Proposition \ref{prop:equispacing}.
\end{remark}

\section{Simulation Results}\label{sec:simulation}

In this section, the proposed method is validated through numerical simulations in MATLAB for straight-path and sinusoidal-curve following. The initial UAV positions are uniformly sampled within a rectangular region bounded by the lines $y=-20 \text{ m}, y=20 \text{ m}, x=20\text{ m}, x=-20 \text{ m}$. The safety distance of the UAVs and activation radius in \eqref{eq:omega_rep} are considered to be $d_{\text{safe}}=0.4$ m and $R_s =1.5 $ m, respectively. The nominal speed $v_{\text{nom}} = 3$ m/s and $\kappa = 1$ m/s, and the gain values $k_r = 11, k_g=0.05, k_{\psi} = 2.3$. The desired uniform spacing $d_{eq}$ is considered as 4 m. Further, we define the inter-UAV distance parameter $\mathcal{E}$ as the minimum distance between any UAV pair at a given time, that is, $\mathcal{E} = \min_{i\neq j} d_{i,j} $.

\begin{figure}[!hbt]
    \centering
    \includegraphics[width=\linewidth]{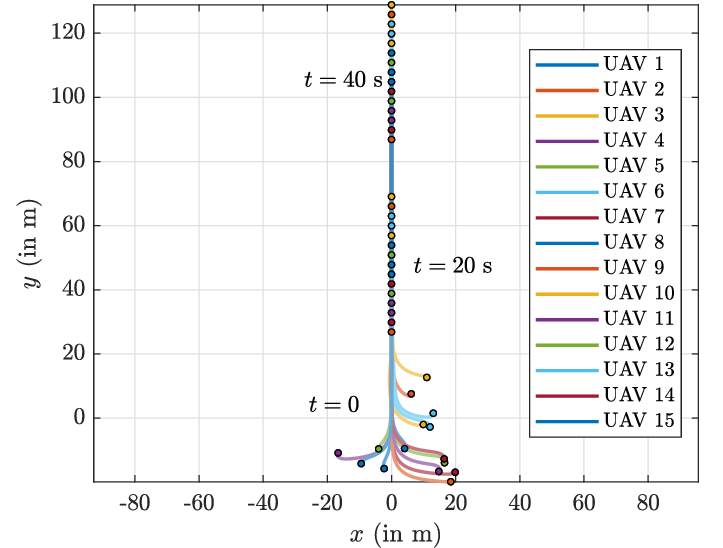}
    \caption{UAV trajectories for straight line path following (U$_i$ represents the plot for the $i$th UAV).}
    \label{fig:st_traj}
\end{figure}

\begin{figure}[!hbt]
    \centering
    \includegraphics[width=\linewidth]{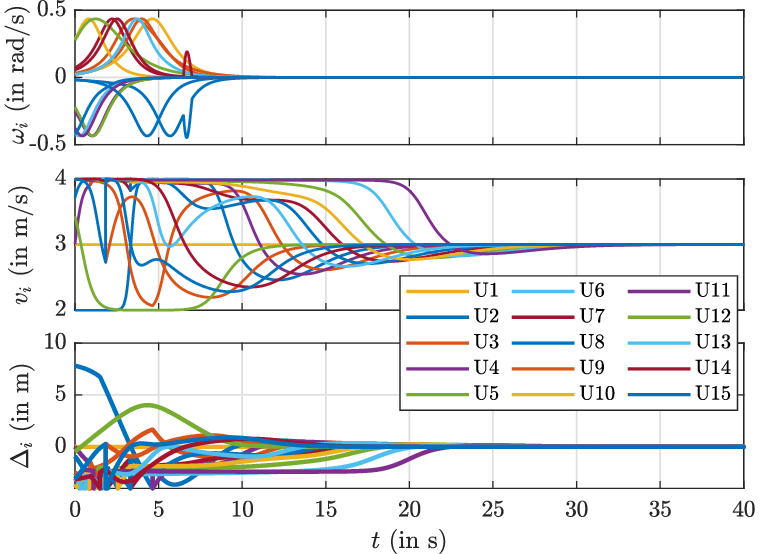}
    \caption{UAV control input and spacing error for straight line path following: angular velocity input (top), linear velocity (middle), and spacing error (bottom) .}
    \label{fig:st_vel}
\end{figure}

\begin{figure}
    \centering
    \includegraphics[width=\linewidth]{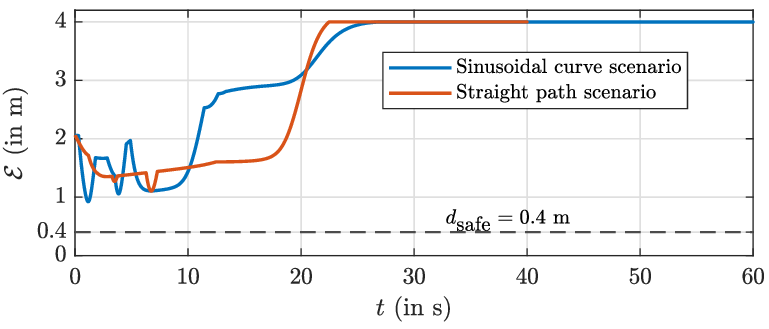}
    \caption{Inter-UAV distance parameter $\mathcal{E}$ for straight line path and sinusoidal curve following.}
    \label{fig:int_uav}
\end{figure}

For the first scenario, straight line path-following is considered $(\mathcal{P}_1: x=0)$. The UAV trajectories are depicted in Fig. \ref{fig:st_traj}, which shows UAV positions at $t = 0, 20$ s and $40$ s. Fig. \ref{fig:st_vel} shows the angular and linear velocity profiles of different UAVs, and the spacing error between consecutive UAVs. The angular velocity for all UAVs converges to 0, while the linear velocity converges to $v_{\text{nom}} = 3$ m/s, while remaining within the bounds $(v_{\text{nom}}-\kappa,v_{\text{nom}}+\kappa)=(2,4)$ m/s. Since $\mathcal{E}>d_{\mathrm{safe}} = 0.4$ m at all times, as shown in Fig. \ref{fig:int_uav}, no collision occurs between UAVs.

\begin{figure}[!hbt]
    \centering
    \includegraphics[width=\linewidth]{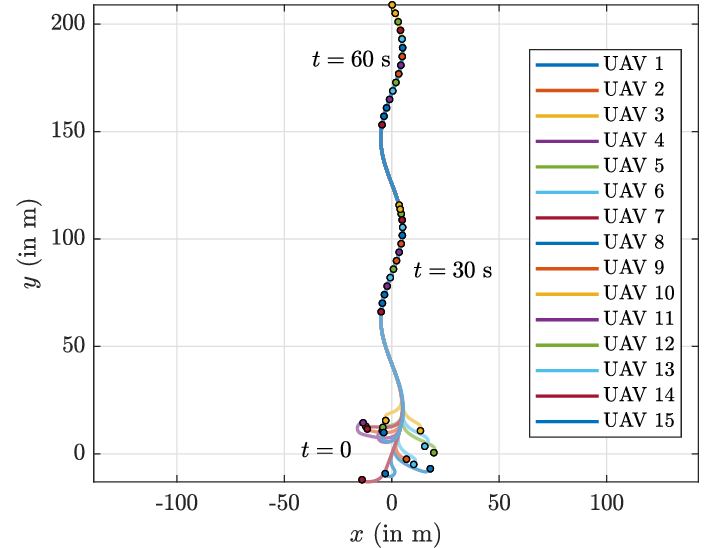}
    \caption{UAV trajectories for sinusoidal curve following.}
    \label{fig:sine_traj}
\end{figure}

In another scenario, illustrated in Fig. \ref{fig:sine_traj}, the UAVs follow a sinusoidal path $\mathcal{P}_2: x = A\sin(ky)$, where $A = 5$ and $k = 0.075$. The angular velocity profile in Fig. \ref{fig:sine_vel} shows a sinusoidal pattern for all UAVs as they converge to the path $\mathcal{P}_2$. Meanwhile, the linear velocity converges to $v_{\text{nom}}$ while remaining within the bounds of $(2, 4)$ m/s. Furthermore, the convergence of the spacing error to zero, as shown in Fig. \ref{fig:sine_vel}, indicates that uniform spacing is achieved among consecutive UAVs. Fig. \ref{fig:int_uav} shows that $\mathcal{E}>d_{\mathrm{safe}}$, indicating no collision occurs.
\begin{figure}[!hbt]
    \centering
    \includegraphics[width=\linewidth]{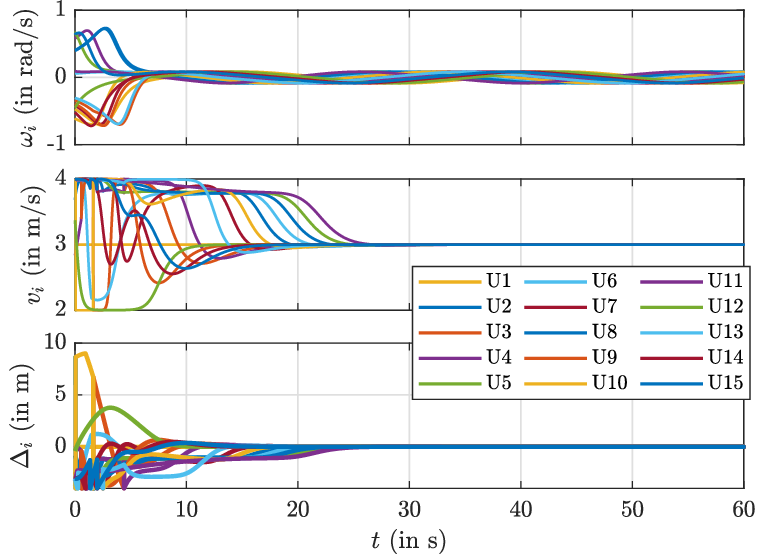}
    \caption{UAV control input and spacing error for sinusoidal curve following: angular velocity input (top), linear velocity (middle), and spacing error (bottom).}
    \label{fig:sine_vel}
\end{figure}

\section{Conclusion}\label{sec:conclusion}
This paper presents a collision avoidance strategy for multiple UAVs following an open reference path while maintaining equal spacing along it. Two types of reference paths, straight line and sinusoidal curve, are considered, and a vector-field guidance law based on an arcsine function of the cross-track error drives each UAV to its reference path. Collision avoidance is achieved using a rotational repulsion mechanism that drives each UAV away from its neighbors based on relative bearing and proximity. The uniform spacing between UAVs along the path is maintained using a linear velocity control law that regulates each UAV's speed based on the spacing error. Numerical simulations with 15 UAVs are performed to show the effectiveness of the proposed method. Future work could explore path-following for more complex open and closed paths. Benchmarking against existing potential function-based methods remains another direction for exploration.

\bibliography{sample}

\end{document}